\documentclass[12pt,letterpaper,top=2.5cm,bottom=2.5cm,left=3cm,right=2.5cm,marginparwidth=1.75cm]{article}
\usepackage{arxiv}

\usepackage{graphicx}      

\usepackage{amsmath} 
\usepackage{amssymb}  
\usepackage{amsthm}

\usepackage{tikz}
\usetikzlibrary{arrows.meta, positioning, calc}

\newtheorem{propp}{Proposition}
\newtheorem{remm}{Remark}
\newtheorem{assm}{Assumption}
\usepackage{booktabs}
\usepackage{array}
\usepackage[table]{xcolor}
\usepackage{caption}
\usepackage{adjustbox}
\usepackage{graphicx}

\usepackage{caption}
\usepackage{subcaption}
\usepackage{textcomp}
\usepackage{url}
\usepackage{colortbl}
\usepackage{soul}
\usepackage{multirow}
\usepackage{float}
\usepackage{physics}
\usepackage{pifont}
\usepackage{color}
\usepackage{alltt}
\usepackage[hidelinks]{hyperref}
\usepackage{enumerate}
\usepackage{siunitx}
\usepackage{breakurl}
\usepackage{epstopdf}
\usepackage{pbox}
\usepackage{stfloats}
\usepackage[vlined,linesnumberedhidden,ruled,resetcount]{algorithm2e}
\usepackage{algpseudocode}  
\usepackage{cite}
\usepackage{amsmath, amsfonts}\usepackage{graphicx}

\usepackage{caption}
\usepackage{subcaption}
\usepackage{textcomp}
\usepackage{url}
\usepackage{colortbl}
\usepackage{multirow}
\usepackage{float}
\usepackage{physics}
\usepackage{pifont}
\usepackage{color}
\usepackage{alltt}
\usepackage[hidelinks]{hyperref}
\usepackage{enumerate}
\usepackage{siunitx}
\usepackage{breakurl}
\usepackage{epstopdf}
\usepackage{pbox}
\usepackage{stfloats}
\usepackage[vlined,linesnumberedhidden,ruled,resetcount]{algorithm2e}
\usepackage{algpseudocode}  
\usepackage{cite}
\usepackage{amsmath, amsfonts}

\allowdisplaybreaks

\usepackage{titlesec}
\usepackage{array}     
\usepackage{diagbox}   
\usepackage{booktabs}  

\begin{document}

\title{Real-Time Regulation of Direct Ink Writing Using Model Reference Adaptive Control} 

\author{Mandana Mohammadi Looey, 
Amrita Basak, 
and Satadru Dey
\thanks{The authors are with the Department of Mechanical Engineering, The Pennsylvania State University, University Park, Pennsylvania 16802, USA. (e-mails: mfm6970@psu.edu, aub1526@psu.edu, skd5685@psu.edu).}
\thanks{This work was supported by National Science Foundation under Grants No. 2346650. The opinions, findings, and conclusions or recommendations expressed are those of the author(s) and do not necessarily reflect the views of the National Science Foundation.}}

\maketitle



\begin{abstract}                

Direct Ink Writing (DIW) has gained attention for its potential to reduce printing time and material waste. However, maintaining precise geometry and consistent print quality remains challenging under dynamically varying operating conditions. This paper presents a control-focused approach using a model reference adaptive control (MRAC) strategy based on a reduced-order model (ROM) of extrusion-based 3D printing for a candidate cementitious material system. The proposed controller actively compensates for uncertainties and disturbances by adjusting process parameters in real time, with the objective of minimizing reference-tracking errors. Stability and convergence are rigorously verified via Lyapunov analysis, demonstrating that tracking errors asymptotically approach zero. Performance evaluation under realistic simulation scenarios confirms the effectiveness of the adaptive control framework in maintaining accurate and robust extrusion behavior.

\end{abstract}



\section{Introduction}


Direct Ink Writing (DIW) is an extrusion‑based 3D printing technique in which a viscoelastic ink is extruded through a nozzle along a programmed tool path onto a build plate. Experimental and numerical studies have shown that printing parameters, including nozzle speed, extrusion rate, and the gap between the nozzle and the build plate, strongly influence buildability and the geometric accuracy of printed features \cite{apsari2025influence}, \cite{imran2023buildability},\cite{scalise2026multiphase} and are critical for process monitoring and control in concrete 3D printing \cite{math11061499}. Achieving stable and consistent ink deposition, as well as the desired geometry and print quality, therefore requires real-time closed-loop control of these parameters \cite{ma14020337}. In this context, model-based adaptive control strategies can compensate for uncertainties and disturbances, enabling in-situ adjustment of process parameters to maintain accurate reference tracking and robust print performance.




Several efforts in the literature have addressed the design of closed-loop controllers and control strategies for extrusion processes. Zomorodi et al. proposed a hierarchical control structure with an explicit model predictive controller that accounts for both extrusion force and ram velocity \cite{zomorodi2016extrusion}. To mitigate under-extrusion and over-extrusion, which lead to inconsistent ink deposition in DIW, Li et al. developed a data-driven reinforcement learning (RL) controller \cite{li2025process}. Kajzr et al. emphasized the need for a dedicated control strategy for 3D concrete printing by implementing an open PLC-based control system \cite{robotics12040096}. Adaptive control approaches have also been explored in extrusion processes. Early work by Guo et al. introduced a fuzzy control scheme with self-adaptation capability to maintain desired product dimensions through real-time monitoring and regulation of polymer extrusion \cite{guo1993polymer}. Perera et al. developed an adaptive controller based on extremum-seeking control to regulate melt pressure via screw rotational speed adjustment \cite{perera2024adaptive}, while Zhao et al. designed a low-order adaptive control scheme to manage extrusion force in the Freeze-form extrusion process \cite{zhao2010adaptive}. Although these studies demonstrate various strategies for achieving target geometries in extrusion-based additive manufacturing, real-time model adaptation and closed-loop control specifically for cementitious extrusion remain largely unexplored.

Motivated by this research gap, this study introduces a generalized adaptive control framework for DIW, demonstrated using cementitious materials as a representative candidate material to account for the inherent uncertainty and trajectory-dependent dynamics of extrusion. The approach is formulated as a model reference adaptive control (MRAC) scheme \cite{narendra2012stable} and follows a mathematical framework similar to that in \cite{dey2015nonlinear}. The proposed control architecture leverages a coupled mathematical model of the extrusion nozzle and build plate motion to represent system dynamics accurately. Thereafter, it employs a parameter adaptation law that updates the feedback-based closed-loop control effort in real time to compensate for fluctuations in build plate speed and inlet mass flow, with the goal of maintaining the desired strand geometry.

Stability of the overall control architecture is established via Lyapunov-based analysis and Barbalat’s lemma, providing theoretical conditions for control and adaptation gains to ensure convergence and robustness. Simulation studies under realistic operating scenarios further demonstrate the effectiveness of the proposed framework in maintaining accurate geometry and consistent extrusion performance. Consequently, this study advances process control in extrusion-based 3D printing of cementitious materials and provides a foundation for the integration of more advanced adaptive and model-based control strategies in large-scale cement printing applications.

The structure of the paper is as follows: Section 2 presents the system model and formalizes the control problem. Section 3 details the adaptive control framework and the control design methodology. Section 4 provides model validation, simulation results, and their analysis. Finally, Section 5 concludes the study and highlights future research directions.

\section{Problem Statement: Control of Extrusion-based 3D Printing Processes}
In this section, we present the basic control architecture, the control-oriented modeling approach, and the formulation of the adaptive control strategy for extrusion-based 3D printing processes.

\subsection{Basic Architecture}
The extrusion-based 3D printing system considered in this study consists of three main sub-systems, as depicted in Fig.~\ref{fig:system_architecture}. Sub-system 1 represents the material feeding mechanism that extrudes the build material. The input to this sub-system is the inlet flow rate, $\dot{m}$, introduced to the nozzle by an external force, and the output is the average flow velocity across the nozzle cross-section at its outlet. Sub-system 2 captures the bending and swelling behavior of the extruded material. Its input is the output of Sub-system 1. This sub-system connects the dynamics of the nozzle flow with the post-deposition dynamics of the strand on the build plate. Sub-system 3 represents the post-deposition behavior of the strand on a build plate moving at velocity $U_s$. Its inputs are the build plate velocity $U_s$ and the output of Sub-system 2. The output is the average flow velocity of the deposited strand on the moving plate.

\begin{figure}[h!]
    \centering
    \includegraphics[width=0.45\textwidth]{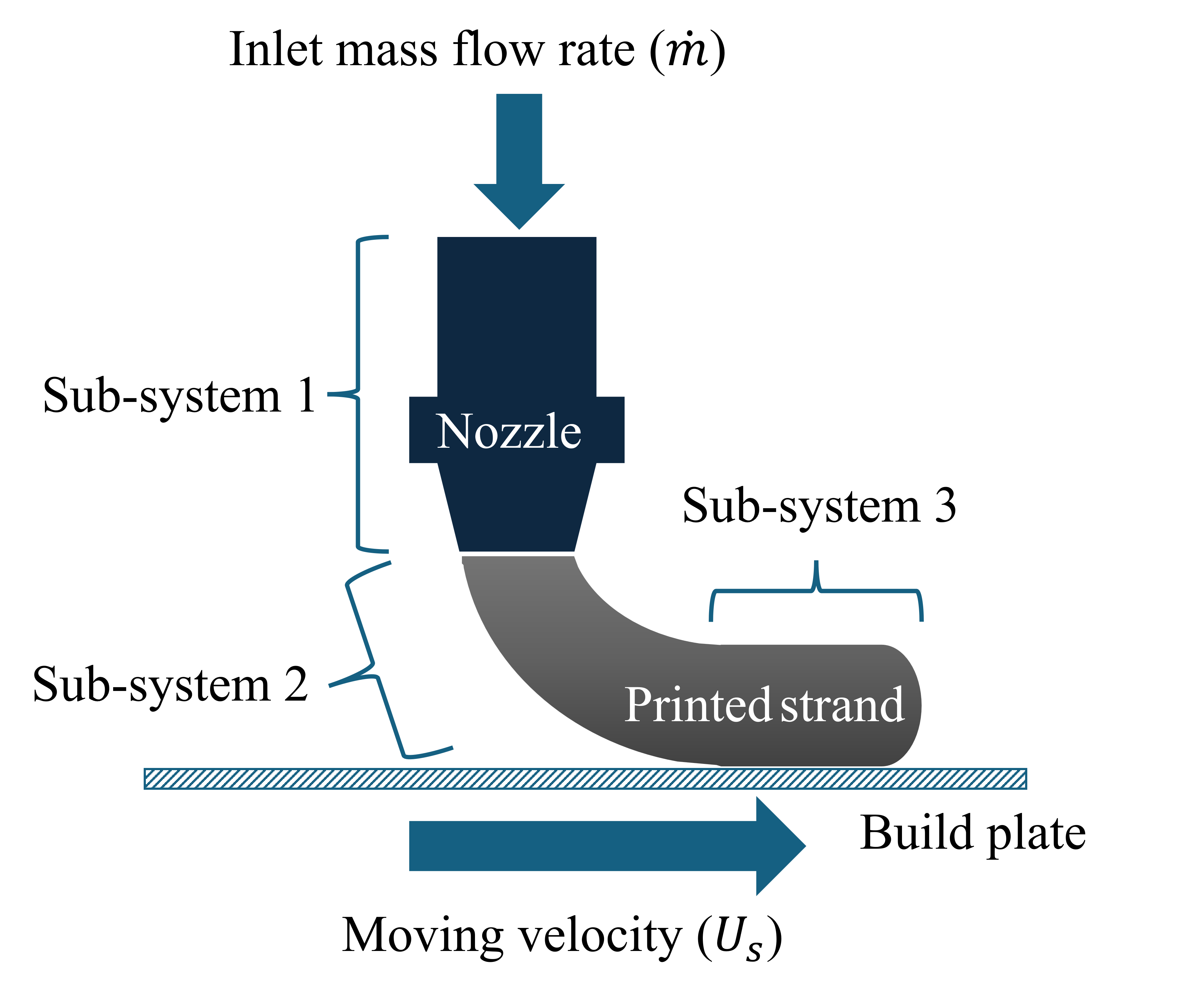}
    \caption{Schematic of the extrusion-based 3D printing process including three sub-systems. }
    \label{fig:system_architecture}
\end{figure}



\subsection{Control-Oriented Model of Extrusion-Based Printing}

We adopt a control-oriented model of the extrusion-based 3D printing system developed in our prior work \cite{looey2025physics}. The model captures the coupled dynamics of material feed, film bending and swelling, and material deposition on the moving build plate. The model is given by:
\begin{align}
&{\dot{\bar{v}}_1}
= \bar{\beta}_{1}\,{p_d}_1
+ \bar{\beta}_{2}\,\bar{v}_1
+ \bar{\beta}_{3}\,\dot{m},  \label{rom-1}\\
&\bar u_{2}= \bar{\beta}_4 \bar v_{1}, \label{rom-2}\\
&{\dot{\bar{u}}_3}
= \bar{\beta}_5 \bar u_2
+\bar{\beta}_6\,\bar{u}_{3}\,+ \bar{\beta}_7\, U_s, \label{rom-3}
\end{align}
where $\bar{v}_1$ is the average flow velocity in the nozzle or sub-system 1, $p_{d_1}$ is the pressure gradient inside the nozzle, $\bar{u}_2$ is the average flow velocity at the outlet of sub-system 2, $\bar{u}_3$ is the output average flow velocity of sub-system 3. Here, the model parameters, $\bar{\beta}$'s are given as functions of the feed rate $\dot{m}$ and plate velocity $U_s$ \cite{looey2025physics}.
\begin{align}
    & \bar{\beta}_1 = \bar{\gamma}_1 \dot{m}, \ \bar{\beta}_2 = \bar{\gamma}_2 \dot{m}, \ \bar{\beta}_3 = \bar{\gamma}_3 \dot{m}, \label{param-1}\\
    & \bar{\beta}_5 = \bar{\gamma}_5 U_s, \ \bar{\beta}_6 = \bar{\gamma}_6 U_s, \ \bar{\beta}_7 = \bar{\gamma}_7 U_s. \label{param-2}
\end{align}


Next, we apply the following reformulation to the model \eqref{rom-1}-\eqref{rom-3}: (i) we incorporate algebraic form of sub-system 2 into the dynamics of sub-system 3, and (ii) instead of input-dependent model parameters, we incorporate constant model parameters and additive uncertainty terms which capture the input-dependence characteristics. Accordingly, the reformulated model becomes:
\begin{align}
&{\dot{\bar{v}}_1}
= \beta_{1}\,{p_d}_1
+ \beta_{2}\,\bar{v}_1
+ \beta_{3}\,(\dot{m} + \Delta_1),  \label{romre-1}\\
&{\dot{\bar{u}}_3}
= \beta_5 \beta_4 \bar v_{1}
+\beta_6\,\bar{u}_{3}\,+ \beta_7\, (U_s + \Delta_3), \label{romre-3}
\end{align}
where $\bar{v}_1$ and $\bar{u}_3$ are now the dynamic states representing the printing system, $\dot{m}$ and $U_s$ are the inputs, and $\Delta_1$ and $\Delta_3$ are the uncertainties.

\vspace{1mm}

\begin{remm}
The uncertainty terms $\Delta_1$ and $\Delta_3$ in \eqref{romre-1}-\eqref{romre-3} mainly account for the input dependent nature of the model parameters. 
This reformulation essentially converts a linear time-varying model into a linear time-invariant model with additive uncertainty, ultimately enabling us to apply linear time-invariant adaptive control design techniques. We also note that $\Delta_1$ and $\Delta_3$ in \eqref{romre-1}-\eqref{romre-3} can potentially capture the effects of other unmodeled dynamics and parametric uncertainties arising from model reduction, modeling assumptions, numerical approximations,  neglected higher-order dynamics, and variations in operating conditions \cite{lopez2016identifying}.
\end{remm}

\subsection{Adaptive Control Problem in Extrusion-Based Printing}
While formulating the adaptive control problem, we first make the following assumptions: 


\vspace{1mm}

\begin{assm}
    The mass flow rate $\dot{m}$ and the plate velocity $U_s$ both serve as control inputs which can be manipulated to achieve the control objective.
\end{assm}

\vspace{1mm}

\begin{assm}
    The velocities $\bar{v}_1$ and $\bar{u}_3$ both are available to be used as feedback signals. This can be achieved either by using on-board sensors or by using real-time state estimators \cite{GREEFF201731}. 
\end{assm}

Under this setting, the adaptive control objective can be written as: \emph{Using sensor feedback (Assumption 2), actuation capabilities (Assumption 1), and the control-oriented model \eqref{romre-1}-\eqref{romre-3} -- design a controller that adapts to uncertainties $\Delta_1$ and $\Delta_3$, and ensures that the dynamical states of the extrusion-based 3D printing system follow the desired velocity trajectories -- ultimately enabling the following physically essential regulations:}
\begin{enumerate}
    \item \textbf{Material flow regulation:} Ensure that the actual material flow rate, which results in the flow velocity $\bar{v}_1$ in the nozzle, tracks the desired flow rate and the consequent desired flow velocity  $v_{r_1}$ generated from a reference model, with minimal error, despite the uncertainties in inlet mass flow rate.
    \item \textbf{Filament velocity regulation:} Maintain the flow velocity on the moving plate ($\bar{u}_3$) close to the desired flow velocity $u_{r_3}$ generated from a reference model, despite the uncertainties in plate movement.
    \item \textbf{Coupled sub-system coordination:} Compensate for interactions between the inlet mass flow rate and the plate velocity through a unified control law that combines nominal and adaptive components.
\end{enumerate}

\section{Adaptive Control Framework for Extrusion-based 3D Printing Processes}


In this section, we discuss the detailed control architecture and the theoretical design of adaptive control algorithm.

\subsection{Adaptive Control Architecture}
Here, we adopt an MRAC approach \cite{narendra2012stable} and the corresponding schematic for the extrusion-based process is given in Fig. \ref{fig:adaptive_control_block_diagram}. The architecture consists of an inner loop and an outer loop. In the inner loop, the error signals between the 3D printer outputs ($\bar{v}_1,\bar{u}_3$) and the model reference signals (${v_r}_1,{u_r}_3$) are computed, which in turn are used by the controller to generate control signals ($\dot{m},U_s$) to be sent to the 3D printer. In the outer loop, reference signals are generated using a reference model, and an adjustment mechanism system (in the form of parameter adaptation law) dynamically updates the control law. 
\begin{figure}[h!]
    \centering
    \includegraphics[width=0.6\textwidth]{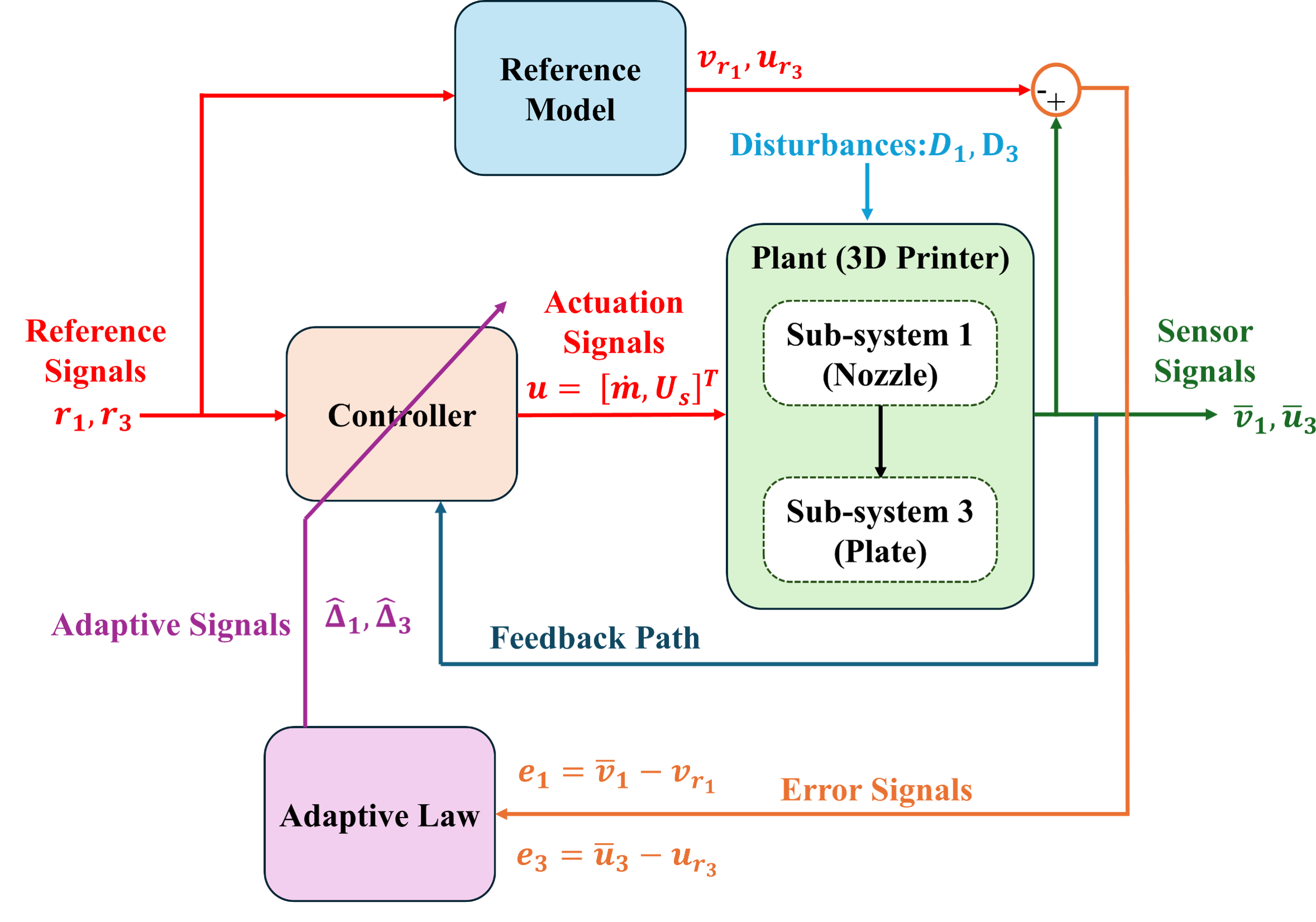}
    \caption{A schematic of model reference adaptive control of extrusion-based 3D printing process.}
    \label{fig:adaptive_control_block_diagram}
\end{figure}

\subsection{Model Reference Adaptive Control Design and Analysis}




The design of MRAC includes designing a feedback control law, a parameter adaptation law, and analyzing the convergence of the closed-loop system under these laws. To this end, we choose the following laws:

\begin{align}
  \begin{split}
    &\text{Control Law }
    \left\{
      \begin{aligned}
        \dot m &= -  k_1 \bar v_1 + r_1 - ({\beta_2}/{\beta_3})p_{d1} - \hat{\Delta}_1,\\ 
        U_s &= -  k_3 \bar u_3 + r_3 - ({\beta_4\beta_6}/{\beta_7})\bar v_1 - \hat{\Delta}_3, \label{sfb-3}
      \end{aligned}
    \right.
  \end{split}
\end{align}
where $k_1$ and $k_3$ are feedback control gains to be designed, $r_1$ and $r_3$ are reference signals, and $\hat{\Delta}_1$ and $\hat{\Delta}_3$ are adaptive parameter estimates from adaptation law given by:
\begin{align}
  \begin{split}
    &\text{Adaptation Law }
    \left\{
      \begin{aligned}
        \dot{\hat \Delta}_1 &= \gamma_1\, p_1\, \beta_3\, (\bar{v}_1-{v_r}_1),\\ 
        \dot{\hat \Delta}_3 &= \gamma_3\, p_3\, \beta_7\, (\bar{u}_3-{u_r}_3), \label{ad-3}
      \end{aligned}
    \right.
  \end{split}
\end{align} 
where $\gamma_1,\gamma_3>0$ are learning rates to be designed, $p_1,p_3>0$ are weight constants to be used later in analysis, and ${v_r}_1,{u_r}_3$ are model reference signals.

Next, the reference model for MRAC is chosen as:
\begin{align}
  \begin{split}
    &\text{Reference Model }
    \left\{
      \begin{aligned}
        & \dot{v}_{r_1} = (\beta_2-\beta_3k_1)v_{r_1} + \beta_3r_1, \\
    & \dot{u}_{r_3} = (\beta_6-\beta_7k_3)u_{r_3} + \beta_7r_3. \label{rm-3}
      \end{aligned}
    \right.
  \end{split}
\end{align} 

In the next proposition, we derive the conditions for achieving the control objectives in terms of bounded error.

\vspace{1mm}

\begin{propp} [Boundedness of reference model tracking and parameter estimation error signals]
Consider the extrusion-based 3D printing system described by \eqref{romre-1}-\eqref{romre-3}, the state feedback control law \eqref{sfb-3}, the parameter adaptation law \eqref{ad-3}, and the reference model \eqref{rm-3}. Then, the reference model tracking error signals $e_1 \triangleq \bar{v}_1 - {v_r}_1$ and $e_3 \triangleq \bar{u}_3 - {u_r}_3$, and the parameter estimation error signals $\tilde{\Delta}_1 \triangleq \Delta_1 - \hat{\Delta}_1$ and $\tilde{\Delta}_3 \triangleq \Delta_3 - \hat{\Delta}_3$ will remain bounded for all time $t\geqslant 0 $, if the following conditions are satisfied:
\begin{align}
    & (\beta_2-\beta_3k_1) < 0, \ (\beta_6-\beta_7k_3) < 0. \label{bc}
\end{align}
\end{propp}
\begin{proof}
Following the definitions of the tracking error signals $e_1 \triangleq \bar{v}_1 - {v_r}_1$ and $e_3 \triangleq \bar{u}_3 - {u_r}_3$, we calculate the dynamics of these error signals as:
\begin{align}
\dot e_1 = & \dot {\bar{v}}_1 - \dot{v}_{r_1} = (\beta_1-\beta_3k_1)e_1 \nonumber\\
& + \beta_3\left(\dot m + \Delta_1 + k_1 \bar v_1 - r_1 + \frac{\beta_2}{\beta_3}p_{d1}\right). \label{e-1}
\end{align}
and
\begin{align}
\dot e_3 = & \dot {\bar{u}}_3 - \dot{u}_{r_3} = (\beta_5-\beta_7k_3)e_3 \nonumber\\
&
+ \beta_7\left(U_s + \Delta_3 + k_3 \bar u_3 - r_3 + \frac{\beta_6\beta_4}{\beta_7}\bar{v}_1\right) \label{e-3}
\end{align}

Next, we also compute the expressions for parameter estimation error dynamics as:
\begin{align}
\dot{\tilde \Delta}_1 &= \dot\Delta_1 - \dot{\hat \Delta}_1 = -\gamma_1 p_1 \beta_3 e_1, \label{d-1}\\ 
\dot{\tilde \Delta}_3 &= \dot\Delta_3 - \dot{\hat \Delta}_3 = -\gamma_3 p_3 \beta_7 e_3. \label{d-3}
\end{align}

Now, we choose a quadratic (hence, positive definite) Lyapunov function candidate ($V$) as:
\begin{align}
V
= \frac{p_1}{2}e_1^2 + \frac{p_3}{2}e_3^2
+ \frac{1}{2\gamma_1}\tilde{\Delta}_1^2 + \frac{1}{2\gamma_3}\tilde{\Delta}_3^2. \label{lyap-0}
\end{align}
where we have by design $p_1,p_3,\gamma_1,\gamma_3>0$. Next, we compute the derivative of the Lyapunov function candidate along the error trajectories:
\begin{align}
\dot V
&= p_1 e_1 \dot e_1 + p_3 e_3 \dot e_3
+ \frac{1}{\gamma_1}\tilde\Delta_1\dot{\tilde\Delta}_1
+ \frac{1}{\gamma_3}\tilde\Delta_3\dot{\tilde\Delta}_3. \label{lyap-1}
\end{align}
Substituting \eqref{e-1}-\eqref{d-3} in \eqref{lyap-1}, we get:
\begin{align}
\dot V
&= p_1 e_1\Bigl\{(\beta_1-\beta_3k_1)e_1
\nonumber\\
&+ \beta_3(\dot m + \Delta_1 + k_1 \bar v_1 - r_1 + \frac{\beta_2}{\beta_3}p_{d1})\Bigr\} \nonumber\\
&+ p_3 e_3\Bigl\{(\beta_5-\beta_7k_3)e_3
\nonumber\\
&+ \beta_7(U_s + \Delta_3 + k_3 \bar u_3 - r_3 + \frac{\beta_6\beta_4}{\beta_7}\bar{v}_1)\Bigr\} \nonumber\\
&+ \frac{1}{\gamma_1}\tilde\Delta_1\left(-\gamma_1 p_1 \beta_3 e_1\right)
+ \frac{1}{\gamma_3}\tilde\Delta_3\left(-\gamma_3 p_3 \beta_7 e_3\right). \label{lyap-2}
\end{align}
Now applying the control laws \eqref{sfb-3} in \eqref{lyap-2}, we get
\begin{align}
\dot{V} &=
p_{1} e_{1}\left\{(\beta_{1}-\beta_{3}k_{1})e_{1}
+\beta_{3}\tilde{\Delta}_{1}\right\} - p_{1}\beta_{3}\tilde{\Delta}_{1}e_{1} \nonumber\\
&
+p_{3} e_{3}\left\{(\beta_{5}-\beta_{7}k_{3})e_{3}
+\beta_{7}\tilde{\Delta}_{3}\right\}
- p_{3}\beta_{7}\tilde{\Delta}_{3}e_{3}. \label{lyap-3}
\end{align}
After canceling the $\tilde{\Delta}$ terms in \eqref{lyap-3}, we have:
\begin{align}
    \dot V
&= p_1(\beta_1-\beta_3k_1)e_1^2 + p_3(\beta_5-\beta_7k_3)e_3^2. \label{lyap-22}
\end{align}
If the conditions \eqref{bc} are satisfied, we have $\dot V \leqslant 0$, which implies $e_1,e_3,\tilde\Delta_1,\tilde\Delta_3$ remain bounded for $t \geqslant 0$ \cite{khalil_nonlinear_2002}.
\end{proof}

Note that \textit{Proposition 1} only proves the boundedness of the error signals. For the asymptotic stability of the error signals, we need further analysis as presented in the next proposition.

\vspace{1mm}

\begin{propp} [Asymptotic stability of reference model tracking and parameter estimation error signals]
Consider the extrusion-based 3D printing system described by \eqref{romre-1}-\eqref{romre-3}, the state feedback control law \eqref{sfb-3}, the parameter adaptation law \eqref{ad-3}, and the reference model \eqref{rm-3}. Then, the conditions \eqref{bc} will also ensure that the tracking and error signals will be asymptotically stable, that is, $e_1,e_3,\tilde\Delta_1,\tilde\Delta_3 \rightarrow 0$ as $t\rightarrow \infty$.
\end{propp}

\begin{proof}
Consider the same Lyapunov function \eqref{lyap-0}. By construction, $V$ has a finite lower bound. Furthermore, \textit{Proposition 1} has already established $\dot V \leqslant 0$. This indicates the existence of the limiting value $\lim_{t\rightarrow \infty} {V}$.

Next, we focus on the expression of $\ddot{V}$, given as:
\begin{align}
    & \ddot{V} = p_1(\beta_1-\beta_3k_1)e_1\dot{e}_1 + p_3(\beta_5-\beta_7k_3)e_3\dot{e}_3. \label{bb-0}
\end{align}
From the expressions of $\dot{e}_1$ and $\dot{e}_3$: 
\begin{align}
    & \dot{e}_1 = (\beta_{1}-\beta_{3}k_{1})e_{1} +\beta_{3}\tilde{\Delta}_{1}, \label{bb-1}\\
    & \dot{e}_3 = (\beta_{5}-\beta_{7}k_{3})e_{3} +\beta_{7}\tilde{\Delta}_{3}, \label{bb-3}
\end{align}
we can conclude that $\dot{e}_1$ and $\dot{e}_3$ are upper bounded since all the variables on the right hand sides of \eqref{bb-1}-\eqref{bb-3} are bounded (as proven in \textit{Proposition 1}). Consequently, all the right hand side terms of \eqref{bb-0} are also upper bounded, ultimately making $\ddot{V}$ upper bounded. Hence, using Barbalat's lemma \cite{slotine1991applied}, we can conclude that $\lim_{t\rightarrow \infty} \dot{V} \rightarrow 0$. Now, following the expression of $\dot V$ in \eqref{lyap-22}, we can conclude that $e_1,e_3 \rightarrow 0$ as $t \rightarrow \infty$ since $\lim_{t\rightarrow \infty} \dot{V} \rightarrow 0$.

Now, we know that the limits $\lim_{t\rightarrow \infty} \dot{e}_1 = -e_1(0)$ and $\lim_{t\rightarrow \infty} \dot{e}_3 = -e_3(0)$ exist. Furthermore, we consider the expressions of the signals $\ddot{e}_1,\ddot{e}_3$:
\begin{align}
    & \ddot{e}_1 = (\beta_{1}-\beta_{3}k_{1})((\beta_{1}-\beta_{3}k_{1})e_{1} +\beta_{3}\tilde{\Delta}_{1}) -\beta_{3}(\gamma_1 p_1 \beta_3 e_1), \label{bb-11}\\
    & \ddot{e}_3 = (\beta_{5}-\beta_{7}k_{3})((\beta_{5}-\beta_{7}k_{3})e_{3} +\beta_{7}\tilde{\Delta}_{3}) -\beta_{7}(\gamma_3 p_3 \beta_7 e_3), \label{bb-13}
\end{align}
Since all the right hand side terms are bounded in the above expressions, we can conclude the boundedness of $\ddot{e}_1,\ddot{e}_3$. Hence, using Barbalat's lemma \cite{slotine1991applied}, we can conclude that $\lim_{t\rightarrow \infty} \dot{e}_1,\dot{e}_3 \rightarrow 0$. Now, consider the expressions \eqref{bb-1}-\eqref{bb-3} where $e_1,e_3,\dot{e}_1,\dot{e}_3 \rightarrow 0$ as $t\rightarrow \infty$. Hence, we can conclude that $\tilde{\Delta}_{1},\tilde{\Delta}_{3} \rightarrow 0$ as $t\rightarrow \infty$.
\end{proof}

\begin{remm}
    In the control design, there are four parameters to be determined: feedback control gains $k_1,k_3$ and parameter adaptation gains $\gamma_1,\gamma_3$. The conditions on the feedback gains are given in \eqref{bc} while the parameter adaptation gains $\gamma_1,\gamma_3$ are chosen to be positive. It is observed that higher the learning rates, faster the convergence of the adaptation. However, such faster convergence is typically associated with high magnitude oscillations in the transient phase. Hence, the parameter adaptation gains are chosen to balance these two factors: faster convergence and acceptable transient oscillations. Similarly, higher values of feedback control gains will also lead faster tracking convergence, but at the cost of amplifying noise. These are chosen to ensure acceptable tracking convergence without significant amplification of noise.
\end{remm}






\section{Results and Discussion}

This section evaluates the performance of the proposed adaptive controller in Section~3. In order to carry out realistic simulations, we utilize a comprehensive \textit{plant model} incorporating input-dependent parameterization (eq. \eqref{param-1}-\eqref{param-2}) as well as additive model uncertainties informed by computational fluid dynamics studies (as performed in \cite{looey2025physics}) of 3D printing scenario under consideration. The model uncertainties are generated from the ensemble probability distribution shown in Fig. \ref{fig:model-unc}, from data collected in \cite{looey2025physics}. Based on the aforementioned setup, we now perform a few representative case studies to evaluate the adaptive controller's performance.

\textbf{Case Study 1 -- Uncertain plate velocity:} 

In order to evaluate the adaptive controller's performance under varying build plate velocity, we assume that plate velocity drops abruptly by a step input disturbance with the size of $40\%$ of the plate velocity obtained from simulation results, and then increased by the same step size after $30$ seconds, which directly affects the flow velocity of sub-system 3. Fig. \ref{fig:adaptive_case1} shows model reference tracking uncertainty and its estimation. The top right sub-figure illustrates the effect of injecting a disturbance on fluctuations in $\bar{u}_3$ and that sub-system 3 is tracking reference model trajectories with small tracking error ($e_3$). Following disturbance injection, after $t=30$ s, the disturbance causes a temporary fluctuation of $\bar{u}_3$. However, fluctuations decrease along time, leading to a perfect overlap of the plant state and the reference state before $t=60$ s. The bottom right plot shows the piecewise-constant disturbance profile after injecting a step disturbance input to the system and the estimation of the disturbance. The estimation ($\hat \Delta_3$) responds to the disturbance injection abruptly with fluctuations oscillating within a range less than the magnitude of the disturbance and approaches the actual value of the disturbance until another disturbance injection at $t=60$ s. The degree of fluctuations and the speed of adaptation depend on the size of the adaptation rate. The top and bottom left plots show that disturbance input does not have any significant effect on the sub-system 1 dynamics compared to the case without disturbance injection.

\begin{figure}[h!]
    \centering
    \includegraphics[width=0.6\textwidth]{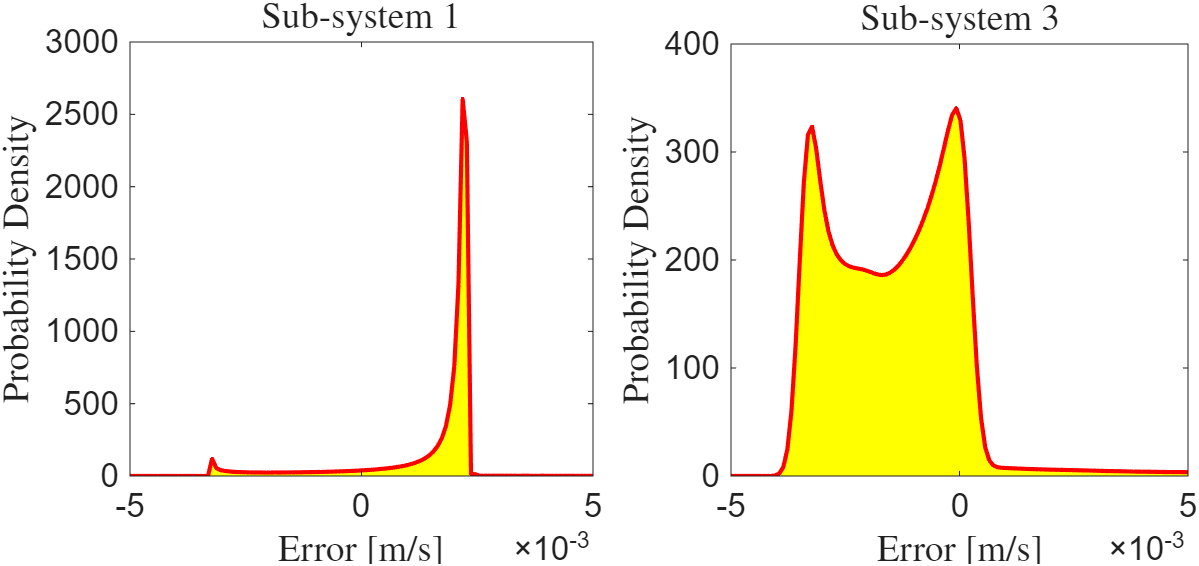}
    \caption{Probability distribution of modeling error.}
    \label{fig:model-unc}
\end{figure}

\begin{figure}[h!]
    \centering
    \includegraphics[width=0.6\textwidth]{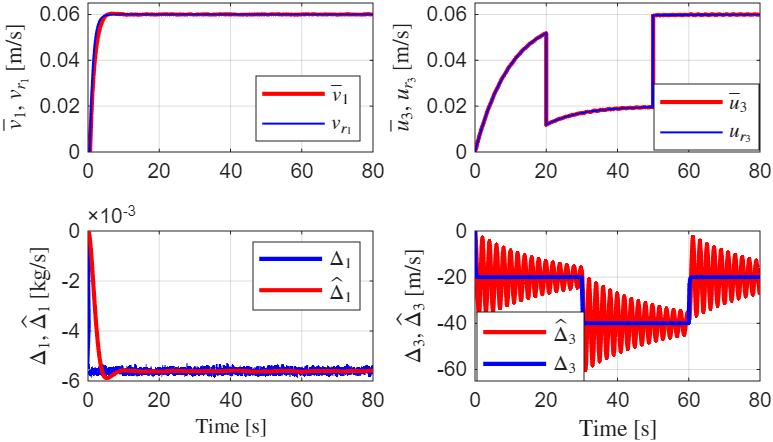}
    \caption{Responses from Case Study 1: Tracking performance under plate velocity disturbance.}
    \label{fig:adaptive_case1}
\end{figure}


\textbf{Case Study 2 -- Uncertain mass flow rate:} 

In this case, a ramp input disturbance with negative slope is applied to sub-system 1 to simulate an undesired decrease of input mass flow rate. Top and bottom left plots in Fig. \ref{fig:adaptive_case2}, demonstrate the effect of the applied disturbance on sub-system 1, while top and bottom right plots indicate the effect of the ramp input disturbance on sub-system 3. The top left plot demonstrates the robustness of adaptive control algorithm to disturbance injection to sub-system 1 by showing perfect tracking of reference trajectories. At the beginning of the ramp, disturbance induces a small and abrupt deviation in the plant state, followed by a more rapid recovery compared to the case, where the disturbance is injected to sub-system 3, driven by the adaptation law and a higher adaptation gain. Similar to Case 1, the injection of a disturbance into sub-system 1 does not affect the dynamics of sub-system 3, as shown in the top and bottom right plots.


\begin{figure}[h!]
    \centering
    \includegraphics[width=0.6\textwidth]{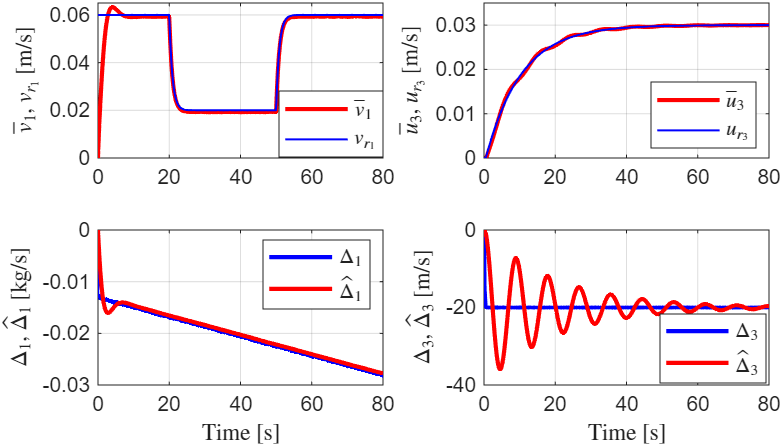}
    \caption{Responses from Case Study 2: Tracking performance under inlet mass flow disturbance.}
    \label{fig:adaptive_case2}
\end{figure}

    


\textbf{Case Study 3 -- Comparison with non-adaptive controller:} 

Here, we compare the adaptive control performance with a non-adaptive controller. The non-adaptive controller is assumed to have knowledge of the disturbances in the beginning of time but does not adapt to the changing nature of the disturbance (that is, the adaptation gains are set to zero, $\gamma_1,\gamma_3 = 0$ in \eqref{ad-3}). For both of these controllers, disturbances are introduced by variations in plate velocity or changes in the mass flow rate. Fig. \ref{fig:non adaptive case 1} shows the responses for the plate-velocity disturbance case by comparing the adaptive controller with the non-adaptive controller for sub-system 1 and sub-system 3. In both sub-systems, the adaptive controller maintains accurate tracking of the reference model even after the disturbance is introduced. In contrast, when adaptation is disabled, the plant response deviates from the reference model.

\begin{figure}[h!]
    \centering
    \includegraphics[width=0.6\textwidth]{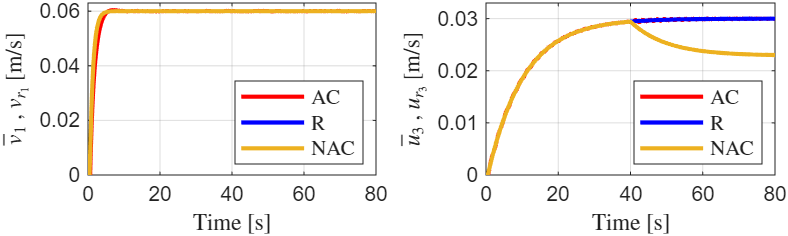}
    \caption{Comparison of adaptive (AC) and non-adaptive (NAC) controllers for a reference (R) under a plate-velocity disturbance.}
    \label{fig:non adaptive case 1}
\end{figure}

\begin{figure}[h!]
    \centering
    \includegraphics[width=0.6\textwidth]{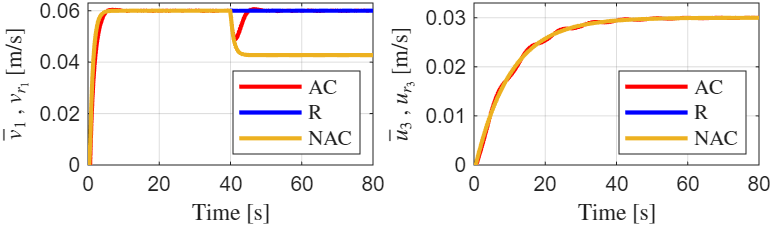}
    \caption{Comparison of adaptive (AC) and non-adaptive (NAC) controllers for a reference (R) under a inlet mass flow disturbance.}
    \label{fig:non adaptive case 2}
\end{figure}

Similarly, Fig. \ref{fig:non adaptive case 2} shows the responses for the case with disturbance to $\dot m$, comparing the adaptive controller with the non-adaptive controller for both sub-systems. Similar to the previous case with plate-velocity disturbance, in this case the adaptive controller shows perfect tracking of the reference model throughout the time, while non-adaptive controller deviates from the reference after disturbance injection. It can be seen in both disturbance cases that the system with the non-adaptive controller stabilizes after a while and remains constant.This behavior demonstrates that the system with a fixed-gain controller is stable and error is bounded; however, there is no adaptation and the system cannot recover after a disturbance injection.  





\textbf{Adaptive control performance under various conditions:} 

Here, we summarize the adaptive control performance under varying magnitudes of disturbances: Cases 1-5 consider a step-like disturbance in the inlet mass flow rate at $t=40$s with various step magnitudes $\delta_m$, while Cases 6-10 consider a step-like disturbance in the plate velocity at $t=40$s with various step magnitudes $\delta_p$. The performance is evaluated by two metrics: \textit{(i) reference tracking convergence time ($t_{c_r}$)} and \textit{(ii) parameter adaptation convergence time ($t_{c_p}$)}. These results are tabulated in Table \ref{tab:ac_perf}. The numbers show a shorter reference tracking convergence time for sub-system 1 compared to sub-system 3, suggesting that the tuned adaptation gains lead to an uniform learning speed and the controller performs faster on the nozzle-side. In contrast, the longer parameter adaptation time for sub-system 3 suggests that the controller stabilizes reference tracking quickly, but the uncertainty estimation converges more slowly, which is because of more complex dynamics of this sub-system. 
\begin{table}[h!]
  \centering
  \caption{Adaptive control performance under mass-flow and plate-velocity disturbances.}
  \label{tab:ac_perf}
  \begin{tabular}{lccc}
    \toprule
    \textbf{Case} & $\delta_m$ [kg/s] & \textbf{$t_{cp}$ [s]} & \textbf{$t_{cr}$ [s]} \\
    \midrule
   
    Case 1  & $+0.0025$ & 7.97  & 4.87 \\
    Case 2  & $-0.0025$ & 8.64  & 4.77 \\
    Case 3  & $-0.0050$ & 8.51  & 5.60 \\
    Case 4  & $-0.0075$ & 8.80  & 5.51 \\
    Case 5  & $-0.0100$ & 8.90  & 5.67 \\
    \midrule
    \textbf{Case} & $\delta_p$ [m/s] & \textbf{$t_{cp}$ [s]} & \textbf{$t_{cr}$ [s]} \\
    \midrule
    Case 6  & $-20$     & 30.97 & 2.11 \\
    Case 7  & $-10$     & 31.00 & 1.80 \\
    Case 8  & $-30$     & 27.50 & 2.34 \\
    Case 9  & $-40$     & 31.00 & 15.40 \\
    Case 10 & $+10$     & 31.00 & 2.43 \\
    \bottomrule
  \end{tabular}
\end{table}

\section{Conclusions}

This paper presents a model-reference adaptive control (MRAC) framework that can robustly regulate nozzle-side flow velocity ($\bar{v}_1$), as well as the flow velocity on the build plate ($\bar{u}_3$) in direct ink writing of cement-based materials under uncertainties and disturbances. The boundedness and asymptotic convergence of reference tracking and parameter estimation errors are guaranteed by a Lyapunov-based design. The realistic disturbances were applied to the inlet mass flow rate and the build plate velocity and a close reference tracking and recovery after disturbance injection proved the high performance of the controller. Comparisons with the non-adaptive controller highlight that although a fixed-gain controller's tracking error is bounded, the controller responses do not follow the reference signal after the disturbance injection due to the lack of adaptation mechanism. Finally, the comparison of the convergence times demonstrated that the parameter adaptation requires longer convergence time due to more complex coupled dynamics of sub-system 3.
\bibliographystyle{ieeetr}
\bibliography{ifacconf}             
                                                   







\end{document}